\documentclass[a4paper,12pt]{article}

\usepackage[utf8]{inputenc}
\usepackage[T1]{fontenc} 

\usepackage[text={6.3in,8.7in},centering]{geometry}

\usepackage{amsmath}
\usepackage{amssymb}
\usepackage{amsthm}

\usepackage{adjustbox}

\usepackage{chngpage}

\usepackage[usenames,dvipsnames,table]{xcolor}
\usepackage{graphicx}
\usepackage{float}
\usepackage{mdframed}
\usepackage{booktabs}

\usepackage[ruled,linesnumbered]{algorithm2e}
\usepackage[pdftex,hidelinks]{hyperref}
\usepackage{kvoptions-patch}
\usepackage[capitalize,nameinlink]{cleveref}
\usepackage[caption=false]{subfig}

\usepackage{authblk}

\graphicspath{{fig/}{./fig/}}

\newtheorem{theorem}{Theorem}
\newtheorem{lemma}{Lemma}

\newcommand{\ch}{\textup{CH}}

\newcommand{\etal}{~et~al.}
\newcommand{\np}{$\mathcal{NP}$-hard}

\newcommand{\p}[2]{\{#1, \ldots, #2\}}
\newcommand{\pij}{\p{i}{j}}

\DeclareMathOperator{\mA}{m_1}
\DeclareMathOperator{\mB}{m_2}

\newcommand{\MinMin}{\textsc{MinMin}}
\newcommand{\MinMax}{\textsc{MinMax}}
\newcommand{\MaxMin}{\textsc{MaxMin}}
\newcommand{\MaxMax}{\textsc{MaxMax}}
\newcommand{\MinMinM}{\textsc{MinMin1}}
\newcommand{\MinMaxM}{\textsc{MinMax1}}
\newcommand{\MaxMinM}{\textsc{MaxMin1}}
\newcommand{\MaxMaxM}{\textsc{MaxMax1}}
\newcommand{\MinMinB}{\textsc{MinMin2}}
\newcommand{\MinMaxB}{\textsc{MinMax2}}
\newcommand{\MaxMinB}{\textsc{MaxMin2}}
\newcommand{\MaxMaxB}{\textsc{MaxMax2}}

\title{New variants of Perfect Non-crossing Matchings%
  \thanks{A preliminary extended abstract was presented at the 36th European Workshop on Computational Geometry (EuroCG) 2020 and a shorter version of this work was presented at the 7th International Conference on Algorithms and Discrete Applied Mathematics (CALDAM) 2021.
    I.M. was partially supported by the Swiss National Science Foundation, project SNF 200021E-154387.
    M.S. was partly supported by the Ministry of Education, Science and Technological Development of the Republic of Serbia (Grant No.\ 451-03-68/2020-14/200125), and the Provincial Secretariat for Higher Education and Scientific Research, Province of Vojvodina.
    H.S. was partially supported by the German Research Foundation, DFG grant FE-340/11-1.
  }
}

\author[1]{Ioannis Mantas}
\author[2]{Marko Savić}
\author[3]{Hendrik Schrezenmaier}

\affil[1]{Faculty of Informatics,
  Università della Svizzera italiana,
  Lugano, Switzerland,
  \normalfont{\texttt{ioannis.mantas@usi.ch}}
}
\affil[2]{Department of Mathematics and Informatics,
  Faculty of Sciences,
  University of Novi Sad, Serbia,
  \normalfont{\texttt{marko.savic@dmi.uns.ac.rs}}
}
\affil[3]{Institut f\"{u}r Mathematik,
  Technische Universit\"{a}t Berlin, Germany,
  \normalfont{\texttt{schrezen@math.tu-berlin.de}}
}

\date{\vspace{-5ex}}

\begin{document}

\maketitle

\begin{abstract}
    Given a set of points in the plane, we are interested in matching them with straight line segments. 
    We focus on perfect (all points are matched) non-crossing (no two edges intersect) matchings.
    Apart from the well known \textsc{MinMax} variation, where the length of the longest edge is minimized, we extend work by looking into different optimization variants such as \textsc{MaxMin}, \textsc{MinMin}, and \textsc{MaxMax}.  
    We consider both the monochromatic and bichromatic versions of these problems and by employing diverse techniques we provide efficient algorithms for various input point configurations.
\end{abstract}


\section{Introduction}

In the \emph{matching problem}, we are given a set of objects and the goal is to divide
the set into pairs such that no object belongs in two pairs.
This simple problem is a classic in graph theory, which has received a lot of attention, both in an abstract and in a geometric setting.
There are plenty of variants of the problem and there is a great plethora of results.

In this paper, we consider the geometric setting where, given a set $P$ of $2n$ points in the plane, the goal is to match points of $P$ with straight line segments, in the sense that each pair of points induces an \emph{edge} of the matching.  
A matching is \emph{perfect} if it consists of exactly $n$ pairs.
A matching is \emph{non-crossing} if all edges induced by the matching are pairwise disjoint.
When there are no restrictions on which pairs of points can be matched, the problem is called \emph{monochromatic}. 
In the \emph{bichromatic} variant, $P$ is partitioned into sets $B$ and $R$ of blue and red points, respectively, and only points of different colors are allowed to be matched.
When $|B|=|R|=n$, the point set $P$ is called \emph{balanced}.

\begin{figure}[t]
	\centering
	\begin{minipage}{.22\textwidth}
	  \centering
	  \includegraphics[width=0.98\textwidth,page=1]{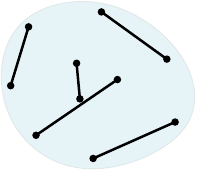}
	\end{minipage}
	\hfill
	\begin{minipage}{.22\textwidth}
	  \centering
	  \includegraphics[width=0.98\textwidth,page=2]{objectiveExamples}
	\end{minipage}
	\hfill
	\begin{minipage}{.22\textwidth}
	  \centering
	  \includegraphics[width=0.98\textwidth,page=4]{objectiveExamples}
	\end{minipage}
	\hfill
	\begin{minipage}{.22\textwidth}
	  \centering
	  \includegraphics[width=0.98\textwidth,page=3]{objectiveExamples}
	\end{minipage}
	\caption{Optimal \MinMinM{}, \MaxMaxM{}, \MinMaxM{}, and \MaxMinM{} matchings of monochromatic points.
	The edges realizing the values of the matchings are highlighted.}	
	\label{fig:objectiveExamples}
\end{figure}

\subsection{Related work on perfect non-crossing matchings}

Geometric matchings find applications in many diverse fields, with the most famous perhaps being operations research, where it is known as the \emph{assignment problem}.
They are useful in the field of shape matching, when shapes are represented by finite point sets, see e.g.,~\cite{veltkamp2001}, and it is a fundamental problem in pattern recognition.
Among others, geometric matchings appear in VLSI design problems, see e.g., \cite{cong1993}, in computational biology, see e.g., \cite{colannino2006},
and are used for map construction or comparison algorithms, see e.g., \cite{eppstein2015}.

In most applications, requiring the matching to be non-crossing or perfect is rather natural.
Given a point set, monochromatic or balanced bichromatic in general position, a perfect non-crossing matching  always exists and it can be found in $O(n\log n)$ time by recursively computing \emph{ham-sandwich cuts}~\cite{lo1994} or by using the algorithm of Hershberger and Suri~\cite{hershberger1992}.
Many times though, it is not sufficient to have any perfect non-crossing matching and the interest lies in finding such a matching with respect to some optimization criterion.

A well-studied optimization criterion is minimizing the sum of lengths of all edges, which we refer to as the \textsc{MinSum} variant. 
This is also known as the \emph{Euclidean assignment} problem or \emph{Euclidean matching}, among others, in the literature.
It is interesting, and not hard to to see, that such a matching is always non-crossing.
For monochromatic point sets, an $O(n^{1.5}\log n)$-time algorithm was given by Varadarajan~\cite{varadarajan1998}.
For bichromatic point sets, Kaplan\etal~\cite{kaplan2020} recently presented an $O(n^2\log^9n\lambda_6(\log n))$-time algorithm,
outperforming previous results~\cite{agarwal2000,vaidya1989}.
When points are in convex position, Marcotte and Suri~\cite{marcotte1991} solved the problem in
$O(n\log n)$ time for both the monochromatic and bichromatic settings.

Another popular goal is to minimize the length of the longest edge, which we refer to as the \textsc{MinMax} variant.
This is also known as the \emph{bottleneck matching} in the literature.
Given a monochromatic point set, it was shown that finding such a matching is \np{} by Abu-Affash\etal~\cite{abu2014}.  This was accompanied by an $O(n^3)$-time algorithm when the points are in convex position.
Recently this was improved to $O(n^2)$ time by Savi{\'c} and Stojakovi{\'c}~\cite{savic2017}.
For bichromatic point sets, Carlsson\etal~\cite{carlsson2015} showed that finding a \textsc{MinMax} matching is also \np{}. 
Biniaz\etal~\cite{biniaz2014} gave an $O(n^3)$-time algorithm for points in convex position and an $O(n\log n)$-algorithm for points on a circle.
These were recently improved to $O(n^2)$ and $O(n)$, respectively, by Savi{\'c} and Stojakovi{\'c}~\cite{savic2018}.

Several other interesting optimization goals have been studied.
For example, in the \emph{fair matching} problem, also known as the \emph{uniform matching}, the goal is to minimize the length difference between the longest and the shortest edge, and in the minimum deviation matching, the difference between the length of the shortest edge and the average edge length should be minimized.
Both were solved in polynomial time by Efrat\etal~ \cite{efrat2001,efrat1996}. 
Alon\etal~\cite{alon1993} studied the \textsc{MaxSum} variant, where the goal is to maximize the sum of edge lengths. They conjectured that the problem is \np{}, and gave an approximation algorithm.

\subsection{Problem variants considered and our contribution.} 
In this work, we continue exploring similar interesting optimization variants in different settings and give efficient algorithms for constructing optimal matchings.
We only deal with perfect non-crossing matchings, so these properties will always be assumed from now on, without further mention.
We consider four optimization variants: $\MinMin$ where the length of the shortest edge is minimized, $\MaxMax$ where the length of the longest edge is maximized, $\MaxMin$ where the length of the shortest edge is maximized, and $\MinMax$.
See \cref{fig:objectiveExamples} for an example of a point set and the four different optimal matchings.

To the best of our knowledge, except for \MinMax{}, the other three variants have not been considered before. 
Studying the \MinMin{} and \MaxMax{} variants is motivated by the analysis of worst-case scenarios for problems where very short or long edges are undesirable, but the selection of edges is not something that we can control.
More generally, the values of \MinMin{} and \MaxMax{} serve as lower and upper bounds on the length of any feasible edge and can be helpful in estimating the quality of a matching, with respect to some objective function.
The \MaxMin{} variant, similar to \MinMax{},
resembles fair matchings in the sense that all edges have similar lengths, analogously to the variants studied in~\cite{efrat2001}.

\begin{table}
    \caption{Summary of results on the optimization of perfect non-crossing matchings.
	The value of the matching can be obtained in the time not indicated with (*).
    The time marked with (*) represents the extra time needed to also return a matching.
	$h$~denotes the size of the convex hull. 
	$\varepsilon$ denotes an arbitrarily small positive constant.
	Results without reference are given in this paper.
	} 
    \label{tab:allResults}

    \begin{adjustbox}{width=\textwidth}

	\setlength{\tabcolsep}{3pt}
	\centering
	\begin{tabular}{@{} lllll @{}}
		\toprule
		\bfseries Monochromatic & \bfseries \MinMinM & \bfseries \MaxMaxM & \bfseries \MinMaxM & \bfseries \MaxMinM\\
		\toprule
		\midrule
		General Position & $O(nh+n\log n)$, & $O(nh) + O(n\log n)^\ast$, & \np~\cite{abu2014}& ?\\
		& $O(n^{1+\epsilon} + n^{2/3}h^{4/3}\log^{3}n)$ & $O(n^{1+\epsilon} + n^{2/3}h^{4/3}\log^{3}n)$ & & \\
		Convex Position & $O(n)$ & $O(n)$ & $O(n^2)$~\cite{savic2017} & $O(n^3)$ \\
		Points on circle & $O(n)$ & $O(n)$ & $O(n)$ & $O(n)$ \\
		\toprule
		\bfseries Bichromatic & \bfseries \MinMinB & \bfseries \MaxMaxB & \bfseries \MinMaxB & \bfseries \MaxMinB\\
		\toprule
		\midrule
		General Position & ? & ? & \np~\cite{carlsson2015}& ?\\
		Convex Position & $O(n)$ & $O(n)$ & $O(n^2)$~\cite{savic2018} & $O(n^3)$ \\
		Points on circle & $O(n)$ & $O(n)$ & $O(n)$~\cite{savic2018} & $O(n^3)$ \\
		Doubly collinear & $O(n)$ & $ O(1) +O(n)^\ast$ & $O(n^4 \log n)$ & ? \\
		\bottomrule
	\end{tabular}
	\end{adjustbox}
\end{table}

We study both the monochromatic and bichromatic versions of these variants.
In the bichromatic version, we assume that $P$ is balanced.
We denote the monochromatic problems with the index~$1$, e.g.~\MinMinM, and the bichromatic problems
with the index~$2$, e.g.~\MinMinB.

These problems are examined in different point configurations.
In \cref{sec:general}, we consider monochromatic points in general position.
In \cref{sec:convex}, points are in convex position.
In \cref{sec:chords}, points lie on a circle.
In \cref{sec:collinear}, we consider \emph{doubly collinear} bichromatic point sets, where the blue points lie on one line and the red points on another line.

\cref{tab:allResults} summarizes the best-known running times for different matching variants including the contributions of this paper.
For each variant we study their structural properties and combine diverse techniques with existing results in order to tackle as many 
configurations as possible.
The various open questions that arise throughout the paper, pave the way for further research in this family of problems.

\section{Monochromatic points in general position}\label{sec:general}

In this section, $P$ is a monochromatic set of points in general position, where we assume that
no three points are collinear.
We denote by $\ch(P)$ the boundary of the convex hull of $P$, by $h$ the number of vertices of $\ch(P)$, by $q_1,\dotsc,q_h$ the counterclockwise ordering of the vertices along $\ch(P)$, and by $d(v,w)$
the Euclidean distance between two points $v$ and $w$.
We call an edge $(v,w)$ \emph{feasible}, if there exists a matching which contains $(v,w)$, and \emph{infeasible} otherwise.

The following lemma gives us a feasibility criterion for an edge $(v,w)$.    

\begin{lemma}
\label{lem:monoGeneralCondition}
    An edge $(v,w)$ is infeasible if and only if (1) $v,w \in \ch(P)$ and (2)~there is an odd number of points on each side of $(v,w)$. 
\end{lemma}

\begin{proof}
    Let $l$ denote the line through the points 
    $v$ and $w$ and let $A,B$ be the subdivision of $P \setminus \{v,w\}$ induced by $l$.
    For the if-part, let $v,w \in   \ch(P)$. Then each edge $(a,b)$ with $a \in A$, $b \in B$ intersects $(v,w)$ and thus cannot be in a matching
    with $(v,w)$. If further $A,B$ have an odd number of points each, at least one point from each set will not be matched if $(v,w)$ is in the matching. Hence, $(v,w)$ is infeasible.
    
    For the only-if-part, let $(v,w)$ be an infeasible edge and suppose that one of (1) and (2) is not fulfilled.
    If (2) is not fulfilled, then $A,B$ have an even number of points each. 
    Therefore, we can find matchings of $A$ and $B$ independently without intersecting $(v,w)$. Therefore, $(v,w)$ is a feasible edge, a contradiction.
    If (1) is not fulfilled, then not both of $v,w$ are in $\ch(P)$.
    So, $l$ crosses at least one edge $(x,y)$ of $  \ch(P)$, with $x \in A, y\in B$, see \cref{fig:general}a.
    But then, both $A \setminus \{ x \}$ and $B \setminus \{ y \}$ contain an even number of points. 
    Thus, there exist matchings of $A \setminus \{ x \}$ and $B \setminus \{ y \}$, which together with $(v,w)$ and $(x,y)$ form a matching of $P$, a contradiction.
\end{proof}

\subsection{\MinMinM{} and \MaxMaxM{} matchings in general position}

The problems $\MinMin$ and $\MaxMax$ are equivalent to finding the \emph{extremal},  shortest or longest, feasible pair.
A main challenge is to check the feasibility of an edge according to \cref{lem:monoGeneralCondition}.
We propose two different approaches.

\paragraph{Using radial orderings.}

The \emph{radial ordering} of a point $p\in P$ is the counterclockwise circular ordering of the points in $P\setminus p$ by angle around $p$.
It is well known that the 
radial orderings of all $p\in P$ can be computed 
in $O(n^2)$ total time using the dual line arrangement of $P$, see e.g.~\cite{handbook-arrangements,handbook-visibility}.

Given a subset $A \subseteq P$, we define the \emph{$A$-weak radial ordering} of a point $p\in P$ as the radial ordering of $p$ where the points from $A$ that occur between two points from $\overline{A}:=P\setminus A$ are given as an unordered set, see Figures \ref{fig:general}b and \ref{fig:general}c.
We are interested in the $\overline{ \ch(P)}$-weak radial orderings of the points in $ \ch(P)$.
These are of interest, as they allow us to check the feasibility of all pairs $(q_i,q_j)$ of points $q_i,q_j \in \ch(P)$ in $O(nh)$ total time using \cref{lem:monoGeneralCondition}.

\begin{figure}[t]
  \centering
  \begin{minipage}{0.32\textwidth}
    \centering
    \includegraphics[width=0.98\textwidth,page=1]{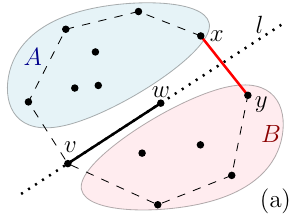}
  \end{minipage}
  \hfill
  \begin{minipage}{0.32\textwidth}
    \centering
    \includegraphics[width=0.98\textwidth,page=2]{inner_extendable}
  \end{minipage}
  \hfill
  \begin{minipage}{0.32\textwidth}
    \centering
    \includegraphics[width=0.98\textwidth, page=3]{inner_extendable}
  \end{minipage}
  \caption{(a) Proof of \cref{lem:monoGeneralCondition}. (b) Proof of \cref{lem:generalMaxMaxSol}. (c) A weak radial ordering.}
  \label{fig:general}
\end{figure}

\begin{lemma}\label{lem:weakRadialOrderings}
    Given a set of points $P$ and a subset $A \subseteq P$ with $|P|=n$ and $|A|=k$,
    the $\overline{A}$-weak radial orderings of all points in $A$ can be computed
    in $O(nk)$ time.
\end{lemma}

\begin{proof}
    First, we compute the dual line arrangement $\mathcal{L}_A$ of $A$ in $O(k^2)$ time~\cite{handbook-arrangements,handbook-visibility}.
    We denote the dual line of a point $p$ by $l_p$.
    For each edge $e$ of $\mathcal{L}_A$, we initialize a set $X_e := \emptyset$, also in $O(k^2)$ total time.
    Then, for each point $p\in P\setminus A$, we find the set $E_p$ of edges of $\mathcal{L}_A$ that are intersected by $l_p$ and add $p$ to all sets $X_e$ with $e\in E_p$. Due to the \emph{zone theorem}~\cite{handbook-arrangements} this takes $O(k)$ time for each $p$.
    
    Finally, we can read off the weak radial ordering of a point $q\in A$ from $\mathcal{L}_A$ and the sets $X_e$ in the following way: Let $p_1,\dotsc,p_{k-1}$ be the ordering of the points in $A\setminus q$ corresponding to the order of intersections of $l_q$ with the other lines in $\mathcal{L}_A$. Further, let $e_i$ be the edge of $\mathcal{L}_A$ between the intersections of $l_q$ with $l_{p_i}$ and $l_{p_{i+1}}$ (with indices understood modulo $k-1$). Then the weak radial ordering of $q$ is $p_1,X_{e_1},p_2,X_{e_2},\dotsc,X_{e_{k-1}}$.
\end{proof}

We use the feasibility criterion of \cref{lem:monoGeneralCondition} and the concept of weak radial orderings to provide algorithms for \MinMinM{} and \MaxMaxM{}.

\begin{theorem}\label{thm:generalmmM}
   If $P$ is in general position, \MinMinM{} can be solved in $O(nh + n\log n)$ time.
\end{theorem}

\begin{proof}
    We initially construct $\ch(P)$ in $O(n\log h)$ time~\cite{chan1996}.
    Then, we compute the $\overline{ \ch(P)}$-weak radial orderings of the points in $ \ch(P)$ in $O(nh)$ total time using \cref{lem:weakRadialOrderings}.
    Now we look for the shortest feasible edge.    

    We first consider edges $(v,w)$ with $v \notin  \ch(P)$ and we want to find $m_1 := \min (\{\, d(v,w) \colon \allowbreak v \in P\setminus \ch(P), w\in P \,\})$.
    By \cref{lem:monoGeneralCondition}, such edges are always 
    feasible.    
    To find $m_1$ we use the fact that the nearest neighbor graph is a subgraph of the Delaunay triangulation, meaning that for each $v \in P$, if $w$ is the nearest point in $P$ to $v$, then $(v,w)$ is an edge of the triangulation.
    We find $m_1$ by finding the shortest edge among all edges incident to some interior point in the Delaunay triangulation.
    This takes $O(n\log n)$ time, as the triangulation can be constructed in $O(n\log n)$ time using standard algorithms and there are $O(n)$ edges to consider.

    Now we consider edges $(v,w)$ with both $v,w \in  \ch(P)$ and we want to find $m_2 := \min (\{\, d(v,w) \colon v,w \in  \ch(P) \,\})$.
    \Cref{lem:monoGeneralCondition} implies that an edge $(q_i,q_{i+1})$ is always
    feasible and that 
    an edge $(q_i,q_{j+1})$ is feasible if and only if 
    (i) $(q_i,q_j)$ is feasible and there is an odd number of points between $q_j$ and $q_{j+1}$ 
    in the radial ordering of $q_i$ or 
    (ii) $(q_i,q_j)$ is infeasible and there is an even number of points between 
    $q_j$ and $q_{j+1}$     
    in the radial ordering of $q_i$.
    Thus, we can find $m_2$ in $O(nh)$ time, using weak radial orderings.
    Hence, we can find the overall minimum $m_{\textnormal{sol}} =\min(m_1,m_2)$, in $O(n\log n + nh)$, concluding the proof.
\end{proof}

It is not hard to observe that the same algorithm but considering the maximum feasible values for $m_1,m_2$ and $m_{sol}$, also solves \MaxMaxM{} in $O(nh + n\log n)$ time.
Using the following lemma we can further improve the time complexity to $O(nh)$. 
    
\begin{lemma}\label{lem:generalMaxMaxSol}
  If $(v,w)$ is a longest feasible edge,
  then one of  $v \in  \ch(P)$ or $w \in  \ch(P)$. 
\end{lemma}

\begin{proof}
    Assume that both $v,w \notin  \ch(P)$. 
    Let $l$ be the line through $v$ and $w$. 
    Then $l$ intersects two edges of $ \ch(P)$.
    Let $(x,y)$ be the edge whose intersection point with $l$ is closer to $w$ than to $u$, see \cref{fig:general}b.
    One of the two angles between $l$ and $(x,y)$ in the interior of $\ch(P)$ is at least $\frac{\pi}{2}$.  
    Let $x$ be the endpoint of $(x,y)$ on this side of $l$. Then $d(v,x) > d(v,w)$.  
    Due to \cref{lem:monoGeneralCondition}, the edge $(v,x)$ is feasible and since it is longer that $(v,w)$, there is a contradiction.
\end{proof}

\begin{theorem}\label{thm:generalMMM}
    If $P$ is in general position, \MaxMaxM{} can be solved in $O(nh)$~time.
\end{theorem}

\begin{proof}
    The algorithm is similar to the \MinMinM{}, described in \cref{thm:generalmmM}, with two changes: the minimizations of $m_1,m_2,m_{sol}$ are replaced by maximizations and 
    for $m_1$ we consider edges $(v,w)$ with $v \in P \setminus  \ch(P)$ and $w\in \ch(P)$. 
    This is sufficient due to \cref{lem:generalMaxMaxSol}.
    This reduces the running time for finding $m_1$ to $O((n-h)h)$ by 
    simply comparing the distances of all $(n-h)h$ pairs of points.
    Hence, the overall running time is reduced to $O((n - h)h + nh) = O(nh)$.    
\end{proof}

\paragraph{Using halfplane range queries.}

Now we take another approach to decide the feasibility of a pair of points from $\ch(P)$. The task of determining the number of points of a given point set lying on one side of a given straight line is known as \emph{halfplane range query} and has been studied extensively over the last decades, see e.g., \cite{agarwal2017}. Using these results to check the criterion of \cref{lem:monoGeneralCondition}, we obtain the following algorithms that are more efficient than those of \cref{thm:generalmmM,thm:generalMMM}, when $h=\Omega(n^c)$ for some constant $c>0$.

\begin{theorem}\label{thm:algHalfplaneRangeQueries}
    Let $P$ be in general position.
    Then \MinMinM{} and \MaxMaxM{} can be solved in $O(n^{1+\epsilon} + n^{2/3} h^{4/3} \log^3 n)$ time where $\epsilon>0$ is an arbitrary constant.
\end{theorem}

\begin{proof}
    We show that the feasibility of all pairs of points of $\ch(P)$ can be decided in the claimed running times. Then, with the aforementioned algorithm and an additional effort of $O(n \log n)$ time, \MinMinM{} and \MaxMaxM{} can be solved.
    
    We distinguish two classes of values of $h$.
    Let $h\leq n^{1/4}$. According to \cite{matousek1993}, halfplane range queries can be answered in $O(n^{1/2})$ time after a preprocessing step costing $O(n^{1+\epsilon})$ time.
    We have to do ${\binom{h}{2}} = O(h^2)$ queries, so the time needed for the queries is $O(h^2 n^{1/2}) = O(n)$. 
    Therefore the preprocessing step dominates the overall time needed, resulting in $O(n^{1+\epsilon})$ total time.
    
    Now let $h \geq n^{1/4}$. We set $m = n^{2/3} h^{4/3}$. 
    Then $n\leq m \leq n^2$ is satisfied,
    which is required by \cite{matousek1993} for the following to hold: Halfplane range queries can be answered in $O(\frac{n}{m^{1/2}} \log^3 \frac{m}{n})$ time after a preprocessing step costing $O(n^{1+\epsilon} + m \log^\epsilon n)$ time. Thus the time needed for the $O(h^2)$ queries is $O(n^{2/3} h^{4/3} \log^3 n)$ and for the preprocessing is $O(n^{1+\epsilon} + n^{2/3} h^{4/3} \log^\epsilon n)$, 
    so $O(n^{1+\epsilon}+ n^{2/3} h^{4/3} \log^3 n)$ time overall.
    Combining the two cases for $h$, the claim follows. 
\end{proof}

Note that given an extremal feasible edge $(v,w)$, we can obtain a matching including $(v,w)$ in $O(n\log n)$ time as follows.
Let the line through $(v,w)$ separate $P$ in two sets $A,B$.
If $A,B$ have an even number of points each, we apply an $O(n\log n)$ time algorithm~\cite{hershberger1992,lo1994} to $A$ and $B$ separately.
If they have an odd number of points, similar to the proof of \cref{lem:monoGeneralCondition}, there exists a feasible edge $(x,y)$ of $\ch(P)$ with $x\in A$ and $y \in B$ and we apply an $O(n\log n)$ algorithm to the sets $A\setminus \{x\}$ and $B\setminus \{y\}$.

\paragraph{Remarks for this section.} 

For the \MinMinM{} problem there is an $\Omega(n\log n)$ lower bound on the time 
complexity, even when $h = O(1)$,
via a reduction from the \emph{closest pair of points problem}: the point set is surrounded by a large triangle or quadrangle (to obtain an even number of points in total) such that all pairs of points of the original point set are feasible and the shortest edge of a \MinMinM{} matching consists of the closest pair of points of the original point set. The complexity of the closest pair of points problem is known to have a lower bound of $\Omega(n\log n)$ via a reduction from the \emph{element distinctness problem}~\cite{benor1983}.

Regarding bichromatic point sets, it remains unknown if it is possible to verify in polynomial time whether a given
edge is feasible or not.
A positive answer would imply polynomial time algorithms for \MinMinB{} and \MaxMaxB{}.
Finally, regarding \MaxMinM{} and \MaxMinB{}, we believe that they are both \np{} problems.

\section{Points in convex position}\label{sec:convex}

In this section, we assume that points in $P$ are in convex position and the counterclockwise ordering of points along $\ch(P)$, $p_0,\dotsc,p_{2n-1}$, is given.
To simplify the notation, we address points by their indices, i.e., we refer to $p_i$ as $i$. Arithmetic operations with indices are done modulo $2n$.
We call edges of the form $(i, i+1)$ \emph{boundary edges} and we call the remaining edges \emph{diagonals}.

Regarding bichromatic point sets, an important concept which captures well the nature of matchings in convex position is the theory of \emph{orbits}~\cite{savic2018}.
More specifically, $P$ is partitioned into orbits.
Each orbit is a a balanced sets of points, and the colors of the points along the boundary of the orbit are alternating,
see~\cref{fig:convex}a.
An important property is that a bichromatic edge $(b,r)$ is feasible if and only if $b$ and $r$ are in the same orbit.

When points are in convex position we can find an arbitrary matching in $O(n)$ time.

\begin{lemma}
    \label{lem:convexMatching}
    If $P$ is in convex position, we can construct an arbitrary matching in $O(n)$ time, both in the monochromatic and bichromatic case.
\end{lemma}

\begin{proof}
    In the monochromatic case, we can trivially pair all boundary edges, e.g.~of the form $(2k,2k+1)$. 
    In the bichromatic case, we first calculate the orbits in $O(n)$ time~\cite
    {savic2018}.
    Then, we choose a color, e.g. red, and in each orbit, we match each red point to the next blue point in the counterclockwise ordering along the boundary of this orbit.
    \cite[Property 20]{savic2018} guarantees that such edges do not cross.
\end{proof}

We first present a general dynamic programming approach and then we give better algorithms for the \MinMin{} and \MaxMax{} variants.
\subsection{A dynamic programming approach}

We can easily solve all four optimization variants in $O(n^3)$ time by a classic dynamic programming approach which is also used in~\cite{abu2014,biniaz2014,carlsson2015} for \MinMax{} problems.
We briefly explain this approach for completeness.

Let $F_i$ be the set of points that point $i$ induces a feasible edge with.
In the monochromatic case, an edge $(i,j)$ is feasible if and only if $i+j$ is an odd number~\cite{abu2014}.
In the bichromatic case $(i,j)$ is feasible if and only if $\pij$ is balanced~\cite{biniaz2014,carlsson2015,savic2018}.

Let $\mA, \mB \in \{\min, \max\}$ be the two optimization functions we use, e.g., if $\mA=\min$ and $\mB=\max$, then we are dealing with the \MinMax{} problem.

The optimal solution restricted to $\pij$ can be recursively expressed as
\[
    M[i,j] = \underset{k \in \{ i+1,j \} \cap F_i}{\mA}  \mB\{d(i,k), M[i+1,k-1], M[k+1,j] \} \enspace .
\]

To determine whether an edge $(i,k)$ is feasible in the bichromatic case, we maintain the number of red and blue vertices encountered while iterating from $i+1$ to $j$. Thus, the total time used for calculating the table $M$ is $O(n^3)$.

\subsection{\MinMin{} and \MaxMax{} matchings in convex position}

We make use of the following two algorithms.
Given two convex polygons $P$ and $Q$, Toussaint's algorithm~\cite{toussaint1984} finds in $O(|P|+|Q|)$ time the vertices that realize the minimum distance between $P$ and $Q$.
Analogously, Edelsbrunner's algorithm~\cite{edelsbrunner1985} finds in $O(|P|+|Q|)$ time the vertices that realize the maximum distance between $P$ and $Q$.

\begin{theorem}  \label{thm:convexmmMandMMM}
    If $P$ is convex, \MinMinM{} and \MaxMaxM{} can be solved in $O(n)$ time.
\end{theorem}

\begin{proof}
    A pair $(i,j)$ is feasible if and only if $i$ and $j$ are of different parity. 
    This suggests that we can split $P$ into two (convex) sets, $P_{\textnormal{odd}}$ and $P_{\textnormal{even}}$,  one containing the even and the other containing the odd indices.
    Then, any edge $(v,w)$ with $v \in P_{\textnormal{even}}$ and $w \in P_{\textnormal{odd}}$ is feasible.
    We now can apply Toussaint's algorithm~\cite{toussaint1984} for \MinMinM{} or Edelsbrunner's algorithm~\cite{edelsbrunner1985} for \MaxMaxM{}.
    All steps can be done in $O(n)$ time.
\end{proof}

\begin{figure}[t]
    \centering
    \begin{minipage}{0.32\textwidth}
        \centering
        \includegraphics[width=0.95\textwidth,page=1]{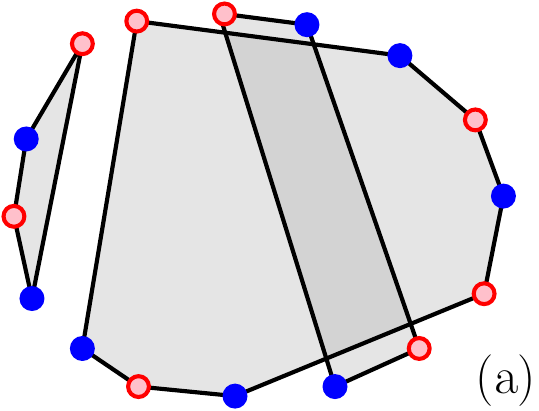}
    \end{minipage}
    \hfill
    \begin{minipage}{0.32\textwidth}
        \centering
        \includegraphics[width=0.95\textwidth,page=2]{convex_algorithm.pdf}
    \end{minipage}
    \hfill    
    \begin{minipage}{0.32\textwidth}
        \centering
        \includegraphics[width=0.95\textwidth,page=3]{convex_algorithm.pdf}
    \end{minipage}
    \caption{
    \MinMinB{} for $P$ in convex position.
    (a) Find orbits. (b) Find the shortest edge between the blue and red polygon of an orbit. (c) Extend to a perfect matching. \label{fig:convex}
    }
\end{figure}

We can obtain the same time complexity for bichromatic point sets, by combining the monochromatic algorithm with the theory of orbits as follows.

\begin{theorem}\label{thm:convexmmBandMMB}
    If $P$ is convex, \MinMinB{} and \MaxMaxB{} can be solved in $O(n)$ time.
\end{theorem}

\begin{proof}
    We first compute all orbits in $O(n)$ time~\cite{savic2018}.
    Due to the alternation of red and blue points along the boundary of the orbits,
    a single orbit can be considered as a set of points in the monochromatic setting, with respect to the feasibility of the edges. 
    
    Thus, for \MinMinB{}, we can select for each orbit, the shortest edge in $O(n)$ time using \cref{thm:convexmmMandMMM}.
    The shortest edge out of all these selected edges is the shortest overall feasible edge of $P$, see~\cref{fig:convex}b.
    Selecting the maximum edges instead solves \MaxMaxB{}.
\end{proof}

We remark that we can construct, in $O(n)$ time, optimal matchings from \cref{thm:convexmmMandMMM,thm:convexmmBandMMB} after finding an extremal feasible edge $(i,j)$, by simply applying \cref{lem:convexMatching} to the sets $\{i+1,\dots,j-1 \}$ and $\{j+1,\dots,i-1 \}$, see~\cref{fig:convex}c.

\nopagebreak

\paragraph{Remarks for this section.}

We believe that $o(n^3)$-time algorithms can be devised for \MaxMin{} by using some ideas similar to those used for \MinMax{} in~\cite{savic2017,savic2018}.

\section{Points on a circle}
\label{sec:chords}

In this section, we assume that all points of $P$ lie on a circle. Obviously, the points are also in convex position, so all the results from~\cref{sec:convex} also apply here.
We present algorithms which either achieve a better time complexity or are significantly simpler.
We use the notation from \cref{sec:convex}.

In addition to the convex position, the results in this section rely on a property of point sets lying on a circle, which we call the \emph{decreasing chords property}. 
A point set $P$ has this property if, for any edge $(i,j)$, it happens that for at least one of its sides, all the possible edges between two points on that side are not longer than $(i,j)$ itself, see \cref{fig:circle}a.

Due to the decreasing chords property, we can easily infer the following.

\begin{lemma}
    \label{lem:shortestIsBoundary}
    Any shortest edge of a matching on $P$ is a boundary edge. \qed
\end{lemma}

\subsection{\MaxMinM{} matching on a circle}

\cref{lem:shortestIsBoundary} suggests an approach for \MaxMin{} problems by \emph{forbidding} short boundary edges and checking whether we can find a matching without them. 

Let some boundary edges be \emph{forbidden} and the remaining be \emph{allowed}. 
A \emph{forbidden chain} is a maximal sequence of consecutive forbidden edges.
A forbidden chain has endpoints $i$ and $j$ if edges $(i,i+1),\dots,(j-1,j)$ are forbidden and edges $(i-1,i)$ and $(j,j+1)$ are allowed, see \cref{fig:circle}b.

\begin{figure}[t]
    \centering
    \begin{minipage}{0.32\textwidth}
        \centering
        \includegraphics[width=0.95\textwidth,page=3]{forbidden_edges.pdf}
    \end{minipage}
    \hfill
    \begin{minipage}{0.32\textwidth}
        \centering
        \includegraphics[width=0.95\textwidth,page=1]{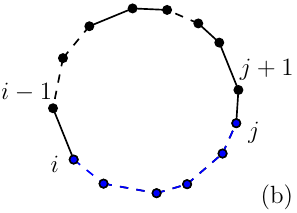}
    \end{minipage}
    \hfill    
    \begin{minipage}{0.32\textwidth}
        \centering
        \includegraphics[width=0.95\textwidth,page=2]{forbidden_edges.pdf}
    \end{minipage}
    \caption{(a) The decreasing chords property. (b) A forbidden chain. (c) Proof of \cref{lem:mmMchains}. (Forbidden edges are shown dashed and allowed edges are shown solid.)
    \label{fig:circle}
    }
\end{figure}

\begin{lemma}\label{lem:mmMchains}
    There exists a matching without the forbidden edges if and only if $l < n$, where $l$ is the length of a longest forbidden chain.
\end{lemma}

\begin{proof}
    If the boundary edges are either all forbidden or all allowed, then the statement trivially holds. So, let us that assume there exists at least one forbidden and at least one allowed boundary edge.    

    Consider a forbidden chain of length $l$ which has endpoints $i$ and $j$. First, we assume that $l \geq n$. Then, at least one matched pair $(a,b)$ has both endpoints in $\pij$. Thus, either $(a,b)$ is a forbidden boundary edge, or it splits $P$ in a way that all points on one side of $(a,b)$ lie completely in $\pij$. So, there exists a matched boundary edge inside $\pij$ and, thus, a matching without forbidden edges does not exist.
    
    Now let us assume that $l < n$. We construct a matching without using forbidden edges with a recursive approach. We match the pair $(i-1, i)$ and consider the set $P' = P \setminus \{i-1, i\}$, see \cref{fig:circle}c. In $P'$, $(i-2,i+1)$ is an allowed boundary edge since it is a diagonal in $P$. We show that $P'$ can be matched by showing that the condition of the lemma holds for $P'$.
    
    Let $i'$ and $j'$ be the endpoints of a longest forbidden chain in $P'$, going counterclockwise from $i'$ to $j'$, and let $l'$ be its length.
    If $l' < n-1$, a matching of $P'$ without forbidden edges can be computed recursively.
    Otherwise, if $l' \geq n-1$, from $l' \leq l < n$ we infer that $l' = l = n-1$.
    Since $l=l'$, the new longest forbidden chain is disjoint from $\{i+1,\dotsc,j\}$, so it is contained in $\{j+1,\dotsc,i-2\}$, see \cref{fig:circle}c. But since $|P'|=2n-2$ and $|\{i+1,\dotsc,j\}|=n-1$, then $|\{j+1,\dotsc,i-2\}| = n-1$ and thus $l' < n-1$, a contradiction.
\end{proof}

\MaxMin{} is equivalent to finding the largest value $\mu$ such that there exists a matching with all edges of length at least $\mu$. Due to \cref{lem:shortestIsBoundary}, it suffices to search for $\mu$ among the lengths of the boundary edges. 
By \cref{lem:mmMchains}, this means that we need to find the maximal length $\mu$ of a boundary edge such that there are no $n$ consecutive boundary edges all shorter than $\mu$. 
An obvious way to find $\mu$ is to employ binary search over the boundary edge lengths and check at each step whether the condition is satisfied or not, which yields an $O(n \log n)$-time algorithm.

A faster approach to finding $\mu$ is as follows. 
Consider all $2n$ sets of $n$ consecutive boundary edges and associate to each set the longest edge in it. 
Then, out of the $2n$ longest edges, we search for the shortest one. This can be done in $O(n)$ time by using a data structure for \emph{range maximum query}, see e.g.~\cite{fischer2011}. 
However, our approach fits under the more restricted \emph{sliding window maximum problem}, for which several simple optimal algorithms are known, see e.g.~\cite{tangwongsan2017}.

\begin{theorem}\label{lem:chordsMmM}
    \label{thm:chordsMmM}
    If $P$ lies on a circle, \MaxMinM{} can be solved in $O(n)$ time.
\end{theorem}

\begin{proof}
    In sliding window maximum algorithm we maintain the \emph{window}, which is a sequence of numbers that we can modify in each step either by adding a new number to it, or by removing the least recently added number, if the sequence is not empty. After each operation the maximum element currently in the window is reported. This is all done in aggregate $O(n)$ time, where $n$ is the number of used numbers.
    
    For our application, we first fill the window with edges $(0, 1), (1, 2), \ldots, \allowbreak (n-1, n)$, in that order, and then, in the $k$-th step ($k$ starting from $0$), we modify the window by removing the edge $(k, k+1)$ and adding the edge $(n+k, n+k+1)$. We repeat this until $k$ reaches $2n-1$, that is, until we go around the circle one full time, asking for the maximum value in the window after each step. Following the previous discussion, the result we are looking for will be the minimum of all these maximums.
\end{proof}

Using \cref{lem:circleMatchingThreshold}, we can also construct an optimal matching within the same time complexity, as the following lemma states.

\newpage

\begin{lemma}\label{lem:circleMatchingThreshold}
    Given a value $\mu>0$, a matching consisting of edges of length at least $\mu$ can be constructed
    in $O(n)$ time if it exists.
\end{lemma}

\begin{proof}
    Since $P$ fulfills the decreasing chords property, the shortest edge (or one of the shortest edges if there is more than one)
    of each matching is a boundary edge. Therefore, a matching consists of edges of length at least $\mu$ if and only if
    the lengths of all its boundary edges are at least $\mu$. We forbid all boundary edges shorter than $\mu$.
    If the longest chain of forbidden boundary edges has length at least $n$, then the wanted matching does not exist due
    to \cref{lem:mmMchains}. Otherwise, we apply the process described in the proof of \cref{lem:mmMchains}. However, we need
    to be careful about how to iteratively find the edge we want to match and remove in each iteration so that the whole
    construction takes only $O(n)$ time.

    First we note that in the process described in the proof of \cref{lem:mmMchains} chains never merge. In each step we reduce the length of exactly one or exactly two chains by $1$. When the length of a chain is reduced to $0$, the chain disappears. So, all the chains maintain their identity throughout the process. For each chain we keep track of its length, and its first endpoint. Also, since we will be removing points, for each point we keep track of the points preceding and following it.

    At the beginning, we calculate the length of all the chains and initialize an array of buckets, where bucket $b$ is a list of all the chains of length exactly $b$. We go through the array starting from the bucket $n-1$ (all buckets above that are empty, by \cref{lem:mmMchains}) and move down towards the bucket $0$. In each step we select the first chain from the current bucket, and check if its length matches the index of the bucket it is in. If it does not, then we move it to the correct bucket, and proceed. If they do match then the selected chain is the chain of the maximum length. If $i$ is the first endpoint of the selected chain, we add the edge between $i$ and the point preceding $i$ to the matching. In constant time we can update everything that we are keeping track of, and continue.

    After we reach bucket $0$, there are no chains left, so we are left with a set of points where all edges are allowed. We can match them in linear time, by matching every second edge.
\end{proof}

\subsection{Other matchings on a circle}

\begin{theorem}\label{thm:convexmMM}
    If $P$ lies on a circle, \MinMaxM{} can be solved in $O(n)$ time. 
\end{theorem}

\begin{proof}
    We show that there exists a \MinMaxM{} matching using only boundary edges. Suppose we have a \MinMaxM{} matching $M$ containing a diagonal $(i, j)$. Assume, without loss of generality, that all edges with endpoints in $\{i, \ldots, j\}$ are at most as long as $(i, j)$. Now we construct a new matching $M'$ by taking all matched pairs in $M$ that are outside of $\pij$ together with edges $(i,i+1), \allowbreak (i+2,i+3),\ldots,(j-1,j)$. The longest edge of $M'$ is not longer than the longest edge of $M$, proving our claim.

    There are only two different matchings made only of boundary edges and in $O(n)$ time we can choose the one with the shorter longest edge.    
\end{proof}

Points on a circle are in convex position, so, both \MinMinM{} and \MinMinB{} can be found in $O(n)$ time using \cref{thm:convexmmMandMMM,thm:convexmmBandMMB}. Instead, we can do it much simpler by finding the shortest feasible boundary edge. By \cref{lem:shortestIsBoundary}, the shortest edge of a matching is a boundary edge in both settings. This can be then extended to a perfect matching using \cref{lem:convexMatching}.

\paragraph{Remarks for this section.}

Results of this section rely only on the decreasing chords property, so they generalize to any point set with this property, e.g.~points on an ellipse of eccentricity at most $1/\sqrt{2}$ or points on a single branch of hyperbola of eccentricity at least $\sqrt{2}$.
Finally, we believe that a combination of the concept of orbits with \cref{lem:shortestIsBoundary} can result in $o(n^3)$-time algorithms for \MaxMinB{} as well.

\section{Doubly collinear points}\label{sec:collinear}

\newcommand{\numP}[1]{\ensuremath{|#1 \cap P|}}
\newcommand{\capP}[1]{\ensuremath{#1 \cap P}}
\newcommand{\edge}[2]{\ensuremath{(#1,#2)}}
\newcommand{\withX}[1]{\ensuremath{(#1_{\rightarrow x \rightarrow})}}
\newcommand{\withoutX}[1]{\ensuremath{{}_{x\cdots}(#1_{\rightarrow})}}
\newcommand{\withXclosed}[1]{\ensuremath{[#1_{\rightarrow x \rightarrow})}}
\newcommand{\withoutXclosed}[1]{\ensuremath{{}_{x\cdots}[#1_{\rightarrow})}}

A bichromatic point set $P$ is \emph{doubly collinear} if the blue points lie on a line $l_B$ and the red points lie on a line $l_R$.
We assume that $l_B$ and $l_R$ are not parallel and that the ordering of the points along each line is given.
Let $x = l_B \cap l_R$ and assume, for simplicity, that $x\notin P$. \Cref{fig:doublyCollinearSetup}a shows this setting.

Let $l\in\{l_B,l_R\}$. Then, for two points $a,b$ on $l$, we denote by $(a,b)$
the open line segment connecting $a$ and $b$. Further, if $a\neq x$, we denote by $\withX{a}, \allowbreak \withoutX{a} \subset l$ the open half-lines starting at $a$ that contain $x$ and do not contain $x$, respectively.
If we replace round brackets by square brackets, e.g.~in $\withoutXclosed{a}$, $(a,b]$, or $[a,b]$, the corresponding endpoint is contained in the set.

Both lines $l_B$ and $l_R$ are split at $x$ into two \emph{half-lines}.
The lines $l_B, l_R$ split the plane into four \emph{sectors}, each bounded by two of the half-lines.
We call a sector \emph{small}, if its angle is acute,
and \emph{big} otherwise.
Note that a point set on the boundary of a big sector has the decreasing chords property.

\begin{figure}[t]
  \centering
    \begin{minipage}{0.32\textwidth}
        \centering
        \includegraphics[width=\textwidth,page=1]{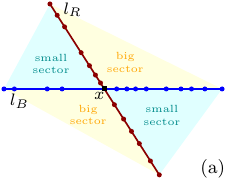}
    \end{minipage}
    \hfill
    \begin{minipage}{0.32\textwidth}
        \centering
        \includegraphics[width=\textwidth,page=2]{doubly_collinear_new}    
    \end{minipage}
    \hfill    
    \begin{minipage}{0.32\textwidth}
        \centering
        \includegraphics[width=\textwidth,page=4]{doubly_collinear_new}
    \end{minipage}
    \caption{(a) A doubly collinear set of points. (b) A feasible edge. (c) An infeasible edge.}
    \label{fig:doublyCollinearSetup}
\end{figure}

The following lemma gives us a feasibility criterion for an edge $\edge{r}{b}$, see Figs.~\ref{fig:doublyCollinearSetup}b and \ref{fig:doublyCollinearSetup}c, which can be checked in $O(1)$ time if the indices of the two points and the position of $x$ are known.
It also indicates an algorithm which, given a feasible edge $\edge{r}{b}$, returns a matching containing $\edge{r}{b}$ in $O(n)$ time.

\begin{lemma}\label{lem:dcFeasibility}
  An edge $\edge{r}{b}$ is feasible if and only if $\numP{(r,x)} \leq \numP{\withX{b}}$ and $\numP{(b,x)} \leq \numP{\withX{r}}$.
\end{lemma}

\begin{proof}
  For the only-if-part, note that, if $r$ and $b$ are matched, then the points in $\capP{(r,x)}$ cannot be matched with the points in
  $\capP{\withoutX{b}}$ since this would yield a crossing with $\edge{r}{b}$. Therefore, the points in $\capP{(r,x)}$ have to be matched with the points in
  $\capP{\withX{b}}$ and, hence, $\numP{(r,x)} \leq \numP{\withX{b}}$. The inequality $\numP{(b,x)} \leq \numP{\withX{r}}$ follows by exchanging the
  colors red and blue in the argument.
  
  For the if-part, we will show that if the two inequalities are satisfied, then we can construct a matching $M$ containing the
  edge $\edge{r}{b}$. We start by adding $\edge{r}{b}$ to $M$. Then we distinguish three cases. If $\numP{(r,x)} = \numP{(b,x)}$, we add the unique non-crossing
  matching of $\capP{(r,x)}$ and $\capP{(b,x)}$ to $M$. Then no edge between two unmatched points intersects an edge of $M$. Hence, we
  can add an arbitrary non-crossing matching of the remaining points to $M$. If $\numP{(r,x)} < \numP{(b,x)}$, let $r'\in \capP{\withX{r}}$ be the
  point such that $\numP{(r,r']} = \numP{(b,x)}$, see \cref{fig:doublyCollinearSetup}b. This point exists since $\numP{(b,x)} \leq \numP{\withX{r}}$.
  We add an arbitrary
  non-crossing matching of $\capP{(r,r']}$ and $\capP{(b,x)}$ to $M$. Then, as in the previous case, no possible edge between the unmatched points
  intersects an edge of $M$ and, hence, we can add an arbitrary non-crossing matching of the remaining points to $M$.
  If  $\numP{(r,x)} > \numP{(b,x)}$, the construction is symmetric to the second case.
\end{proof}

\subsection{\MinMinB{} and \MaxMaxB{} matchings in doubly collinear point sets}

Let $l_R'$ and $l_B'$ be a red and a blue half-line, respectively.
The following lemma is a direct consequence of \cref{lem:dcFeasibility} and allows us to find for each point in $\capP{l_R'}$ the closest point in $\capP{l_B'}$ it induces a feasible edge with
in $O(n)$ total time.

\begin{lemma}\label{lem:dcCanditatesInLinear}
  Let $r\in R$ and $r' \in \capP{\withoutX{r}}$. Let $b,b' \in \capP{l_B'}$ be closest to $r$ and $r'$, respectively, such that $\edge{r}{b}$ and $\edge{r'}{b'}$ are feasible.
  Then $b'\in \capP{\withoutXclosed{b}}$.
\end{lemma}

\begin{proof}
  Suppose that $b' \in \capP{\withX{b}}$. Since $r'\in \capP{\withoutX{r}}$, we have $\numP{(r,x)} < \numP{(r',x)}$. Since $\edge{r'}{b'}$ is a feasible edge,
  we have $\numP{(r',x)} \leq \numP{\withX{b'}}$ by \cref{lem:dcFeasibility}. Together, this gives $\numP{(r,x)} < \numP{\withX{b'}}$.
  Similarly, we obtain $\numP{(b',x)} < \numP{\withX{r}}$. Therefore, due to \cref{lem:dcFeasibility}, $\edge{r}{b'}$ is a feasible edge.
  Analogously, we can show that $\edge{r'}{b}$ is a feasible edge.
  
  Since $\edge{r}{b}$ and $\edge{r'}{b'}$ are crossing edges, we have $d(r,b') + d(r',b) < d(r,b) + d(r',b')$. Hence, $d(r,b') < d(r,b)$
  or $d(r',b) < d(r',b')$. Since $\edge{r}{b'}$ and $\edge{r'}{b}$ are feasible edges, the former is a contradiction to the minimality of $\edge{r}{b}$
  and the latter is a contradiction to the minimality of $\edge{r'}{b'}$.
\end{proof}

\begin{theorem}\label{thm:doublymmBNew}
  If $P$ is doubly collinear, \MinMinB{} can be solved in $O(n)$ time.
\end{theorem}

\begin{proof}
  Using \cref{lem:dcCanditatesInLinear}, we can find for each pair $l_R',l_B'$ of half-lines of $l_R$ and $l_B$, respectively,
  the closest points of all points of $\capP{l_R'}$ in $\capP{l_B'}$ and the closest points of all points
  of $\capP{l_B'}$ in $\capP{l_R'}$ they induce a feasible edge with in $O(n)$ time. This gives a linear number of candidate edges for the shortest feasible edge
  and then the minimum can be computed in $O(n)$ time.
\end{proof}

With the following lemma we can find a longest feasible edge in $O(1)$ time. We call a point $p\in P$ an \emph{extremal} point
if $\numP{\withoutX{p}} = 0$.
    
\begin{lemma}\label{lem:dcLongest}
  The longest edge between points in $R$ and $B$ is realized by a pair of extremal points.
\end{lemma}

\begin{proof}
  Let $\edge{r}{b}$ be a longest edge between $R$ and $B$ and assume that one of them, say $b$, is not an extremal point.
  The distance function between $r$ and $l_B$ is monotone increasing in on of the two direction from $b$.
  Let $b'$ be the next point of $B$ on $l_B$ after $b$ in this direction. Then $d(r,b') > d(r,b)$, in contradiction to the
  assumption that $\edge{r}{b}$ is a longest edge.
\end{proof}
    
\begin{theorem}\label{thm:doublyMMB}
 If $P$ is doubly collinear, \MaxMaxB{} can be solved in $O(1)$ time.
\end{theorem}

\begin{proof}
  Using \cref{lem:dcLongest}, we obtain at most four candidates for the longest edge.
  Since all of them are in the boundary of the convex hull of $P$, they are all feasible.
  Thus, we only have to find their maximum in $O(1)$ time.
\end{proof}

\subsection{\MinMaxB{} and \MaxMinB{} matchings in doubly collinear point sets}

\paragraph{One-sided case.}
Here we consider the \emph{one-sided} doubly collinear case,
where all red points are on the same side of $l_B$, see \cref{fig:oneSidedDoublyCollinear}. Since in this case the extremal red point must be matched with one of the two extremal blue points, all four optimization variants can be solved in $O(n^2)$ time by dynamic programming.

\begin{theorem}\label{thm:osdcGeneralDP}
  If $P$ is one-sided doubly collinear, \MinMaxB{} and \MaxMinB{} can be solved in $O(n^2)$ time.
\end{theorem}

\begin{proof}
  Let $B^-$ and $B^+$ be the sets of blue points on the two half-lines of $l_B$. We label the points of $B^-$ as $b^-_1, \ldots, b^-_{|B^-|}$, the points of $B^+$ as $b^+_1, \ldots, b^+_{|B^+|}$, and the red points as $r_1, \ldots r_n$, all three enumerations starting from the point nearest to $x$ and going away from it, see \cref{fig:oneSidedDoublyCollinear}a.

  \begin{figure}
    \centering
    \begin{minipage}{0.49\textwidth}
        \centering
        \includegraphics[width=0.9\textwidth,page=1]{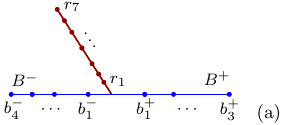}
    \end{minipage}
    \hfill
    \begin{minipage}{0.49\textwidth}
        \centering
        \includegraphics[width=0.9\textwidth,page=2]{one_sided_doubly_collinear}  
    \end{minipage}
    \caption{(a) A one-sided doubly collinear point set. (b) An example of a matching $M_k$.}
    \label{fig:oneSidedDoublyCollinear}
  \end{figure}

  We present a dynamic programming algorithm which solves all optimization problems, \MinMinB{}, \MinMaxB{}, \MaxMinB{}, and \MaxMaxB{}, for the one-sided case. Let $\mA, \mB \in \{\min, \max\}$ be the two optimization functions we use, e.g., if $\mA=\min$ and $\mB=\max$, then we are dealing with the \MinMaxB{} problem.

  Let $S_{n^-,n^+}$ be the optimal solution to the $\mA\mB$ problem where we are restricted to the first $n^-$ points of $B^-$, the first $n^+$ points of $B^+$, and the first $n^-+n^+$ red points. The red point farthest from $x$ must be connected to one of the two extremal blue points, which gives us the recurrence
  \[ S_{n^-,n^+} = m_1 \{m_2\{d(b^-_{n^-},r_{n^-+n^+}), S_{n^--1,n^+}\}, m_2\{d(b^+_{n^+},r_{n^-+n^+}), S_{n^-,n^+-1}\}\} \enspace . \]

  Using this formula we can compute the solutions to all subproblems in $O(n^2)$ time. Constructing an optimal matching is easily done in $O(n)$ time once we have all $S_{n^-,n^+}$ values.
\end{proof}

To construct a faster algorithm for \MinMaxB{}, we notice that there exists (also in the two-sided case) an optimal matching of a special form. It can be obtained from an arbitrary optimal matching by applying local changes that do not change the objective value.

\begin{lemma}\label{lem:minMaxOptMatchingStructure}
  There exists an optimal matching for \MinMaxB{} of the following form.
  For each half-line $l'$, the points of $\capP{l'}$ that are matched in the small incident sector are consecutive points, see \cref{fig:dcSpecialMatchingExamples}c.
\end{lemma}

\begin{proof}
  Let $M$ be an optimal \MinMaxB{} matching.
  Let $l_R'$ be one of the red half-lines. Let $r$ be a point in $\capP{l_R'}$, let~$b$ be the point $r$ is matched with, and assume that $\edge{r}{b}$ lies in a small sector.
  Further, let $r_1,r_2 \neq r$ be two consecutive points in $(x,r)$ with $r_1$ closer to $x$ than $r_2$.
  Let $b_1,b_2$ be the points that $r_1,r_2$ are matched with in $M$, respectively, and assume that $\edge{r_1}{b_1}$ lies in a small sector and $\edge{r_2}{b_2}$ in a big sector.
  Since $r_1$ and $r_2$ are consecutive points, replacing $\edge{r_1}{b_1},\edge{r_2}{b_2}$ in $M$ by $\edge{r_1}{b_2},\edge{r_2}{b_1}$ does not produce any crossing edges.
  We will show that this replacement also does not increase the length of a longest edge of the matching.
  Because of the decreasing chords property of the big sector, $d(r_1,b_2) < d(r_2,b_2)$.
  For the edge $\edge{r_2}{b_1}$ we distinguish three cases, see \cref{fig:dcLocalChangeMinMax}.
  (i) If the angle $b_2b_1r_2$ is obtuse, then $d(r_2,b_1) < d(r_2,b_2)$.
  (ii) If the angle $b_1r_2r_1$ is obtuse, then $d(r_2,b_1) < d(r_1,b_1)$.
  (iii) If both angles are acute or right, then $d(r_2,b_1) < d(r,b)$.
  
  \begin{figure}[H]
    \centering
    \begin{minipage}{0.32\textwidth}
        \centering
        \includegraphics[width=\textwidth,page=1]{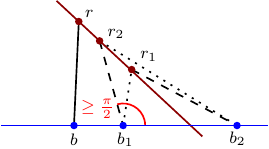}
    \end{minipage}
    \hfill
    \begin{minipage}{0.32\textwidth}
        \centering
        \includegraphics[width=\textwidth,page=2]{local_change_minmax_dc.pdf}  
    \end{minipage}
    \hfill    
    \begin{minipage}{0.32\textwidth}
        \centering
        \includegraphics[width=\textwidth,page=3]{local_change_minmax_dc.pdf}
    \end{minipage}
    \caption{The three cases of the proof of \cref{lem:minMaxOptMatchingStructure}.}
    \label{fig:dcLocalChangeMinMax}
  \end{figure}
  
  By repeatedly applying these swaps we can transform $M$ into an optimal matching that fulfills the desired property on the half-line $l_R'$. If in $M$ the property was already fulfilled on another half-line, it is still fulfilled there in the new matching. Hence, repeating this swapping process on the other half-lines yields an optimal matching of the desired form.
\end{proof}

\begin{figure}[t]
  \centering
  \begin{minipage}{0.32\textwidth}
    \centering
    \includegraphics[width=0.9\textwidth,page=3]{doubly_collinear_special_examples}
  \end{minipage}
  \hfill
  \begin{minipage}{0.32\textwidth}
    \centering
    \includegraphics[width=0.9\textwidth,page=1]{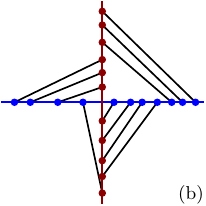} 
  \end{minipage}
  \hfill    
  \begin{minipage}{0.32\textwidth}
    \centering
    \includegraphics[width=0.9\textwidth,page=2]{doubly_collinear_special_examples}
  \end{minipage}
  \caption{Special types of optimal matchings (a)~ for \MinMaxB{},
    (b)~ for \MinMaxB{} and \MaxMinB{} if $\alpha=\frac{\pi}{2}$, and
    (c)~ for \MinMaxB{} if $\alpha < \frac{\pi}{4}$.}
  \label{fig:dcSpecialMatchingExamples}
\end{figure}

\begin{theorem}\label{thm:osdcMinMax}
  If $P$ is \emph{one-sided doubly collinear}, \MinMaxB{} can be solved in $O(n \log n)$ time.
\end{theorem}

\begin{proof}
    We use the notation of the proof of \cref{thm:osdcGeneralDP}.
    Suppose $|B^-| > 0$ and $|B^+| > 0$ (otherwise there is only one possible matching, easily constructed in $O(n)$ time). Without loss of generality, suppose that the sector incident to $B^+$ is big, i.e., $\angle r_1 x b^-_1 \leq \angle b^+_1 x r_1$.
    
    For $k \in \{1, \ldots, n-|B^-|+1\}$, let $M_k$ be the matching comprised of the small sector edges $b^-_i r_{k+i-1}$, $i \in \{1, \ldots, |B^-|\}$, and with all other points matched in the only possible way through the big sector, see \cref{fig:oneSidedDoublyCollinear}b. By \cref{lem:minMaxOptMatchingStructure} there is an $s$ such that $M_s$ is an optimal matching.
    
    Note that, for all $k \in \{1, \ldots, n-|B^-|\}$, the longest edge of $M_k$ in the big sector is $b_{|B^+|}^+ r_n$ and that the longest edge of $M_{n-|B^-|+1}$ in the big sector is $b_{|B^+|}^+ r_{n-|B^-|}$.
    For a matching $M$, we denote by $M'$ the subset of $M$ consisting of the edges in the small sector. We set $\max(M') = \max \{\, d(v,w) \colon (v,w)\in M' \,\}$.
    Let $k^\ast \in \{1, \ldots, n-|B^-|\}$ be the index such that $max(M_{k^\ast}')$ is minimal. Then $M_{k^\ast}$ and $M_{n-|B^-|+1}$ are the only candidates for the optimal \MinMaxB{} matching.
    
    Let $j$ be such that $b^-_j r_{k+j-1}$ is the longest edge in $M'_k$. If $\angle b^-_j r_{k+j-1} x$ is obtuse then $\max(M'_k) = d(b^-_j, r_{k+j-1}) < d(b^-_j, r_{l+j-1}) \leq \max(M'_l)$ for all $l < k$. If the angle is acute, the same holds for all $l > k$. This observation allows us to find the $k \in \{1, \ldots, n-|B^-|\}$ minimizing $\max(M'_k)$ using binary search in $O(|B^-| \log (n-|B^-|+1)) = O(n \log n)$ time.
\end{proof}

\paragraph{General case.}

We return to general doubly collinear matchings and consider the \MinMaxB{} problem.
By only considering matchings as described in \cref{lem:minMaxOptMatchingStructure}, enumerating all possible choices for the decision which blue point is matched through which sector,
and applying \cref{thm:osdcMinMax}
for the two resulting one-sided subproblems,
we obtain the following result.

\begin{theorem}\label{thm:doublymMB}
  If $P$ is doubly collinear, \MinMaxB{} can be solved in $O(n^4 \log n)$ time.
\end{theorem}

\begin{proof}
  It suffices to consider matchings of the form described in \cref{lem:minMaxOptMatchingStructure}.
  Let $l_B'$ be one of the blue half-lines.
  We iterate through all $O(n)$ choices of how many points of $l_B'$ are matched through the incident small sector.
  This choice implies for the three other half-lines how many of their points are matched through their incident small sector.
  Then, for each of the blue half-lines, we iterate through all $O(n)$ choices of consecutive subsets of their points of the desired size for the set of points that are matched through the incident small sector.
  Finally, for each of these $O(n^3)$ cases, we apply the algorithm of \cref{thm:osdcMinMax} to the two one-sided subproblems with only one red half-line in $O(n \log n)$ time. Taking the minimum of the optimal solutions of all these cases yields the optimal solution of the original problem in total $O(n^4 \log n)$ time.
\end{proof}
    
\paragraph{Special angles of intersection.}

Finally, we consider special values for the angle of intersection $\alpha$ of $l_B$ and $l_R$.
By proving the existence of optimal matchings of a special form, we derive improved algorithms for these cases.

\begin{lemma}\label{lem:orthogonalOptMatchingStructure}
  If $\alpha=\frac{\pi}{2}$, for \MinMaxB{} and \MaxMinB{} there exist optimal
  matchings of the following form. Each of the four half-lines is cut into two parts and all points in such a part are matched to points of the same half-line, see \cref{fig:dcSpecialMatchingExamples}b.
\end{lemma}

\begin{proof}
  Let $M$ be an optimal \MinMaxB{} matching.
  Let $l_R'$ be one of the blue half-lines. Let $r$ be the point in $\capP{l_R'}$ farthest from $x$.
  Further, let $r_1,r_2 \neq r$ be two consecutive points in $\capP{l_R'}$ with $r_1$ closer to $x$ than $r_2$.
  Let $b,b_1,b_2$ be the points that $r,r_1,r_2$ are matched with in $M$, respectively.
  Assume that $b_1,b$ are on the same blue half-line and $b_2$ is on the other blue half-line.
  \Cref{fig:doublyCollinearOrthogonalProof}a shows this situation.
  Then swapping the edges $\edge{b_1}{r_1},\edge{b_2}{r_2}$ in $M$ with the edges $\edge{b_1}{r_2},\edge{b_2}{r_1}$ does not increase the length of a
  longest edge of $M$ since $d(b_1,r_2) < d(b,r)$ and $d(b_2,r_1) < d(b_2,r_2)$.
  Therefore, the matching is still an optimal \MinMaxB{} matching after the swap.
  By repeatedly applying these swaps we can transform $M$ into an optimal \MinMaxB{} matching
  that fulfills the desired property on the half-line $l_R'$. If in $M$ the property was already
  fulfilled on another half-line, it is still fulfilled there in the new matching. Hence,
  repeating this swapping process on the other half-lines yields an optimal \MinMaxB{} matching
  of the desired form. The proof for \MaxMinB{} is analogous.
\end{proof}

\begin{theorem}\label{thm:orthogonalLinearTime}
  If $\alpha=\frac{\pi}{2}$, \MinMaxB{} and \MaxMinB{} can be solved in $O(n)$ time.
\end{theorem}

\begin{proof}
  We only consider matchings of the form as described in \cref{lem:orthogonalOptMatchingStructure}.
  We will show that there are only $O(n)$ matchings of this form and that the value of each of these matchings
  can be determined in $O(1)$ time.
    
  Each half-line is split into two parts. We decide for each half-line independently the points of which of these
  two parts are matched with the points of which of the two half-lines of the other color. There are $2^4 = O(1)$ possibilities.
  If we now decide for one of the four half-lines between which two points to split it, the matching is fixed (the choice
  might turn out to be infeasible). There are at most $n+1 = O(n)$ possibilities for this split.
  Since, for each sector, the shortest edge is the edge closest to $x$ and the longest edge
  is the edge farthest from $x$, the value of this matching can be found by computing and comparing the lengths of at most $4 = O(1)$ edges.
\end{proof}

\begin{figure}[t]
  \centering
  \begin{minipage}{0.32\textwidth}
    \centering
    \includegraphics[width=0.9\textwidth,page=1]{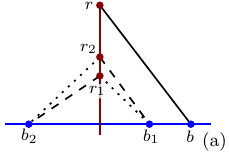}
  \end{minipage}
  \hfill
  \begin{minipage}{0.32\textwidth}
    \centering
    \includegraphics[width=0.9\textwidth,page=2]{doubly_collinear_special_proof}
  \end{minipage}
  \hfill    
  \begin{minipage}{0.32\textwidth}
    \centering
    \includegraphics[width=0.9\textwidth,page=3]{doubly_collinear_special_proof}
  \end{minipage}
  \caption{(a) The situation of the proof of \cref{lem:orthogonalOptMatchingStructure}.
    (b) The situation of the proof of \cref{lem:smallAngleOptMinMaxMatchingStructure} with $\beta \leq \frac{\pi}{2}$.
    (c) The situation of the proof of \cref{lem:smallAngleOptMinMaxMatchingStructure} with $\beta \geq \frac{\pi}{2}$.}
  \label{fig:doublyCollinearOrthogonalProof}
  \label{fig:doublyCollinearSmallAngleSmallBetaProof}
  \label{fig:doublyCollinearSmallAngleLargeBetaProof}
\end{figure}

\begin{lemma}\label{lem:smallAngleOptMinMaxMatchingStructure}
  If $\alpha \leq \frac{\pi}{4}$, there exists an optimal \MinMaxB{}
  matching of the following form. Each of the four half-lines is split into an inner and an outer part, where the inner part is closer to $x$, and points
  from an inner (outer) part are matched through a big (small) sector, see \cref{fig:dcSpecialMatchingExamples}c.
\end{lemma}

\begin{proof}
  Let $M$ be an optimal \MinMaxB{} matching.
  Let $l_R'$ be one of the blue half-lines. Let $r_1,r_2$ be two consecutive points in $\capP{l_R'}$ with $r_1$ closer to $x$ than $r_2$.
  Let $b_1,b_2$ be the points that $r_1,r_2$ are matched with in $M$, respectively.
  Assume that the edge $\edge{b_1}{r_1}$ lies in a small sector and that the edge $\edge{b_2}{r_2}$ lies in a big sector.
  Let $\beta < \pi$ be the angle between the line segments $[r_2,b_1]$ and $[r_2,x]$.
  If $\beta$ is acute (\cref{fig:doublyCollinearSmallAngleSmallBetaProof}b shows this situation), we have $d(r_2,b_1) \leq d(r_2,x) \leq d(r_2,b_2)$. For the first inequality, we use that $\alpha \leq \frac{\pi}{4}$.
  If $\beta$ is obtuse (\cref{fig:doublyCollinearSmallAngleLargeBetaProof}c shows this situation), we have $d(r_2,b_1) \leq d(r_1,b_1)$.
  Further, we have $d(r_1,b_2) < d(r_2,b_2)$.
  Therefore, swapping the edges $\edge{b_1}{r_1},\edge{b_2}{r_2}$ in $M$ with the edges $\edge{b_1}{r_2},\edge{b_2}{r_1}$ does not increase the length of a
  longest edge of $M$.
  Since $r_1,r_2$ are consecutive points, the swap also does not produce crossing edges.
  Hence, the matching is still an optimal \MinMaxB{} matching after the swap.
  By repeatedly applying these swaps we can transform $M$ into an optimal \MinMaxB{} matching
  that fulfills the desired property on the half-line $l_R'$. If in $M$ the property was already
  fulfilled on another half-line, it is still fulfilled there in the new matching. Hence,
  repeating this swapping process on the other half-lines yields an optimal \MinMaxB{} matching
  of the desired form.
\end{proof}

\begin{theorem}
  If $\alpha \leq \frac{\pi}{4}$, \MinMaxB{} can be solved in $O(n)$ time.
\end{theorem}

\begin{proof}
  We only consider matchings of the form as described in \cref{lem:smallAngleOptMinMaxMatchingStructure}. We will show that there are only $O(n)$ matchings of this form and that the value of each of these matchings can be computed in $O(1)$ time after doing some precomputing in $O(n)$ time.
  
  The matchings described in \cref{lem:smallAngleOptMinMaxMatchingStructure} are in particular of the form of the matchings
  described in \cref{lem:orthogonalOptMatchingStructure}. Therefore, it follows as in the proof of \cref{thm:orthogonalLinearTime} that there are only $O(n)$ matchings of this form. Let $b_1,b_2,\dotsc$ and $r_1,r_2,\dotsc$ be the blue and red points, respectively, on the boundary of one of the two small sectors in the order of decreasing distance to $x$. If a matching has $k$ edges in this sector, then the longest edge in this sector has the value $m_k := \max_{i=1,\dotsc,k} d(b_i,r_i)$. Since $m_{k+1} = \max ( m_k, d(b_{k+1},r_{k+1}) )$ for all $k\geq 1$, the values $m_k$, $k=1,2,\dotsc$, can be computed in total $O(n)$ time. The same can be done for the other small sector. Since the big sectors have the decreasing chords property, the longest edge in a big sector is always the edge farthest from $x$. Therefore, after precomputing the values $m_k$ for both small sectors, the value of a matching can be found in $O(1)$ time.
\end{proof}

\paragraph{Remarks for this section.}
The presented algorithms for \MaxMinB{} and \MinMaxB{} rely on the existence of an optimal matching with a special structure: the points on each half-line of at least one color are partitioned into $O(1)$ subsets of consecutive points and all points of the same subset are matched through the same sector. Without any special structure it is difficult to make any assumptions, as for example for the \MaxMinB{}, for which we are currently not aware of any such structure.

\section{Concluding remarks}

We considered new optimization variants for perfect non-crossing matchings of points in the plane.
In most \MinMin{} and \MaxMax{} variants, we came up with optimal algorithms by exploiting structural properties of the point sets, combined with existing techniques from diverse problems.
On the contrary, we saw that the \MaxMin{} variant exhibits a significant difficulty.
Designing efficient algorithms even for simple configurations, as cocircular or doubly collinear points, is not at all obvious and thus, each variation is quite interesting on its own.
Throughout~the paper we posed several open questions together with suggestions for approaches.
For instance, regarding convex bichromatic point sets, can orbits help to improve the \MaxMin{} algorithms?
Regarding arbitrary point sets, is there a polynomial time feasibility check for a bichromatic edge? 
Are the \MaxMin{} variants \np{} as their \MinMax{} counterparts?
It would be interesting to see how \cref{tab:allResults} can be filled with optimal time algorithms or hardness results.

\paragraph{Acknowledgements.}

Preliminary discussions were held during the \href{https://dccg.upc.edu/irp2018/}{Intensive Research Program in Discrete, Combinatorial and Computational Geometry} which took place in Barcelona in 2018.
We are grateful to the \href{http://www.crm.cat}{Centre de Recerca Matem\`{a}tica}, Universitat Autónoma de Barcelona,for hosting the event and to the organizers for providing the platform to meet and collaborate. 
We would also like to thank Carlos Alegr\'{i}a, Carlos Hidalgo Toscano, Oscar Iglesias Valiño, and Leonardo Mart\'{i}nez Sandoval for initial discussions on the problems, and Carlos Seara for raising a question that motivated this work.
Finally, we would like to thank  an anonymous reviewer for bringing to our attention the halfplane range queries.

\newpage

\bibliography{references}
\bibliographystyle{abbrv}

\end{document}